\documentclass[journal]{IEEEtran}
\usepackage[document]{ragged2e}
\usepackage{cite}
\usepackage{amsmath,amssymb,amsfonts}
\usepackage{algorithmic}
\usepackage{enumerate}
\usepackage{comment}
\usepackage{pifont}
\usepackage{mathtools}
\usepackage{chngcntr}
\usepackage{amsthm}
\usepackage{lettrine}
\usepackage{graphicx}
\usepackage{hyperref}
\usepackage{color,soul}
\usepackage{textcomp}
\usepackage{xcolor}
\usepackage{float}
\usepackage[labelsep=period]{caption}
\renewcommand{\thetable}{\Roman{table}}
\usepackage[utf8]{inputenc}
\usepackage[english]{babel}
\newtheorem{theorem}{Theorem}
\newcommand{\believes}{\mid\equiv}
\newcommand{\sees}{\triangleleft}
\newcommand{\oncesaid}{\mid\sim}
\newcommand{\controls}{\Rightarrow}
\newcommand{\fresh}[1]{\#(#1)}
\newcommand{\combine}[2]{{\langle #1 \rangle}_{#2}}
\newcommand{\encrypt}[2]{{ \{ #1 \} }_{#2}}
\newcommand{\sharekey}[1]{\xleftrightarrow{#1}}
\newcommand{\pubkey}[1]{\xmapsto{#1}}
\newcommand{\secret}[1]{\xleftrightharpoons{#1}}
\DeclareMathAlphabet\mathbfcal{OMS}{cmsy}{b}{n}

\renewcommand\thesubsection{\Alph{subsection}}
\begin{document}
\title{Secure Lightweight Authentication for Multi User IoT Environment}
\author{Chintan Patel,~\IEEEmembership{Member,~IEEE}
\IEEEcompsocitemizethanks{\IEEEcompsocthanksitem Chintan Patel is a Post doctorate Fellow in University of Sheffield, UK.\protect\\
E-mail: chintan.p592@gmail.com}}
\IEEEtitleabstractindextext{%
\begin{abstract}
\justifying{
The Internet of Things (IoT) is giving a boost to a plethora of new opportunities for the robust and sustainable deployment of cyber-physical systems. The cornerstone of any IoT system is the sensing devices. These sensing devices have considerable resource constraints, including insufficient battery capacity, CPU capability, and physical security. Because of such resource constraints, designing lightweight cryptographic protocols is an opportunity. Remote User Authentication ensures that two parties establish a secure and durable session key. This study presents a lightweight and safe authentication strategy for the user-gateway (U-GW) IoT network model. The proposed system is designed leveraging Elliptic Curve Cryptography (ECC). We undertake a formal security analysis with both the Automated Validation of Internet Security Protocols (AVISPA) and Burrows–Abadi–Needham (BAN) logic tools and an informal security assessment with the Delev-Yao channel. We use publish/subscribe based Message Queuing Telemetry Transport (MQTT) protocol for communication. Additionally, the performance analysis and comparison of security features show that the proposed scheme is resilient to well-known cryptographic threats. }
\end{abstract}
}

\maketitle

\IEEEdisplaynontitleabstractindextext
\IEEEpeerreviewmaketitle
\ifCLASSOPTIONcompsoc
\IEEEraisesectionheading{\section{INTRODUCTION}\label{sec:introduction}}
\else
\section{\textbf{INTRODUCTION}}
\fi
\justifying
\noindent \lettrine[findent=2pt]{\textbf{E}}{ }xpansion of internet-based services significantly impacts routine life of the people. Due to its ambient and easy-to-use support services, the adoption of Internet of Things (IoT) based services (i.e., smart electricity supply, smart agricultural, smart amenity distribution, industry 4.0,  smart healthcare, smart home) has soared in recent years.

According to recent Gartner forecasts, more than 25.44 billion IoT gadgets will be incorporated into the cyber world's working environment by 2030. IoT-based smart health care connects healthcare entities such as patients, doctors, hospital administrators, drug suppliers, ambulances, and pharmacists using a network of smart healthcare devices. Doctors can obtain real-time health data from patients via wearable health devices using an intelligent patient monitoring system.

IoT based \textbf{smart grid} provides a intelligent system for electricity generation, electricity distribution and customer payments \cite{12farhangi}. The smart grid aims to convert current client-server type electricity distribution to peer-to-peer type energy distribution. Using a two-way electricity line, consumers can also sell their self-generated energy to power distribution companies. The smart meter provides real-time monitoring for the energy distribution system \cite{32Zafar}. 

IoT-based \textbf{intelligent and smart agriculture} helps farmers in environment monitoring, soil quality analysis, fertilizer requirement analysis, water distribution control, and so on. Farmers can save water, money, and time by using smart agriculture applications. Farmers can get inputs from agriculture scientists and the government about their fertilizer needs, and new crop decreases using the real-time data generated from the ambient deployment of nano-sensors on farms \cite{4Antonacci}.

An IoT-based \textbf{smart home} provides automation and smart control for their ambient home devices. Smart home users can control home devices from anywhere globally using internet-connected intelligent home appliances. The smart home user can do the security checkup, thermal monitoring, light controlling, washing machine controlling, and many more smart activities using in-hand mobile devices. 

The vast data generated from these IoT-based services is stored in the cloud and is used for generating knowledge through data processing. Thus, IoT provides an opportunity for the real-time monitoring of sensor data for quick decision-making and storage of data for data analysis and futuristic planning. Every IoT based data service pass through the following four basic layers:
\begin{itemize}
    \item The \textbf{\textit{ground layer}} provides the physical deployment of sensing devices and actuator devices on the ground. This layer works as a data generator.
    \item The \textbf{\textit{second layer}} consists of gateway devices that collect data from ground layer sensing devices and deliver this data to the end-user as well as to the cloud for any further processing. 
    \item The \textbf{\textit{third layer}} provides data analysis using machine learning techniques and artificial intelligence-based intelligent decision making.
    \item The \textbf{\textit{fourth layer}} focuses on end-users who collect the data from the third layer as well as the ground layer (based on application) via second layer devices.
\end{itemize}
The IoT deployment comes with many challenges and opportunities. Starting from deploying thousands of tiny and heterogeneous devices on the ground level to data collection, data analysis, and intelligent decision-making forms significant challenges for researchers. 
The recent past surveys by \cite{28Whitemore} show that \textit{standardization, communication protocol designing, data analysis, security, and privacy} are the highly notable challenges in the IoT setup. These surveys also highlight that security is a big concern, among others. Some of the significant security issues are data confidentiality, user privacy, device authentication, physical security, and so on \cite{1alaba}. Followings are two crucial reasons behind the need for a full-proof security mechanism for the IoT system. 
\begin{itemize}
\item Numerous heterogeneous tiny resource constraint devices on the ground level.
\item Privacy of the user's data.  
\end{itemize}

We can divide the IoT devices among three major parts:
\begin{itemize}  
\item \textbf{User devices} in the IoT are a combination of resource constraints and resource-capable devices. Resource constraint devices like wearable devices to resource-capable devices like laptops and mobiles.
\item \textbf{Gateway devices} in the IoT are considered resource-capable devices and can work as intermediary devices between sensing devices and user devices. They provide setup support to the other devices in the IoT system. The IoT user communicates with the gateway device to receive runtime sensing data.
\item \textbf{Sensing devices} in the IoT involve tiny sensors and actuator devices. They are not capable of performing any traditional and non-traditional security mechanisms. Many security papers highlight that sensing devices do cryptographic operations, but it is non-practical in real time due to energy issues and on-time service requirements. e.g., A thermal sensor deployed in the house can not run for the cryptographic operation when you ask for temperature. They transmit data to the home gateway device, and the user receives data from it.   
\end{itemize}
Thus, the communication between sensors, gateway devices, and end-users must be secure enough to fulfill all the critical security goals. The secure authentication mechanism between these devices provides secure key generation for communication and mutual trust among each communicating party. A secure authentication system fulfills most major security goals except some like access control \cite{3ammar}, \cite{35rupa}.

Cryptography is a branch of mathematics that deals with enumerations and executions of cryptosystems. . With the run-up towards smart technology, the need for lightweight cryptography came into the picture. The reason behind this need is the use of numerous resource constraint devices in the deployment of sensor-based IoT application deployments. These resource constraint devices are short of computation memory and storage capability. It is nearly impossible for these devices to perform sixteen rounds of Data Encryption Standard (DES) and exponential computations of the RSA promptly and without higher energy utilization. In 2003, Hankerson et al. highlighted  \textit{Elliptic Curve Cryptography (ECC)}, which is much lighter than traditional crypto methods in computations and storage requirements \cite{13hankerson}. Due to its appealing features such as reduced key sizes, relatively short time requirements, and limited resource utilisation, the ECC became a well-known cryptographic approach for resource constraint devices.  . The \textit{Elliptic Curve Diffie-Hellman encryption (ECDHE)} provides a lighter version of Diffie-Hellman with \textit{Elliptic Curve Discrete Logarithm Problem (ECDLP)}.

\textbf{\textit{Contributions : }}For the U-GW-based network model presented in section \ref{sec:networkmodel}, we provide a secured and coherent two-factor authentication protocol employing a password and a smart card (SC). In the U-GW-based network model, we consider that the IoT application user communicates with the gateway device for receiving the sensor data. Over here, we believe that the sensing device deployed in the local network are tiny devices, and they do not perform security mechanisms. In this model, the IoT user receives sensor data from the gateway device through a precarious channel. We provide a rigorous security analysis of the proposed scheme using globally recognized tools such as BAN Logic and AVISPA. An informal security analysis using widely adopted Dolev-Yao attack model is also provided. We also compare the provided solution to other approaches and analyze performance of proposed work. This examination demonstrates the proposed work's uniqueness, efficiency, and reliability. 

\textbf{\textit{Motivations : }}The strong \textit{motivation} for this paper is the existence of numerous vulnerabilities in the recently proposed schemes. After a thorough literature study, we critically observed that it is challenging to design a lightweight authentication security scheme for sensing-based applications. These devices require an authentication scheme that uses less energy, less time, and lower space for the secure key exchange; hence IoT users can timely receive heterogeneous IoT sensor data from the gateway device. In the proposed scheme, the secure key exchange between the gateway and the IoT user assures secure data transmission of sensing data among IoT user and gateway device over the anxious public internet. Another strong motivation for proposing this work is a synchronous implementation of the presented work. We implemented the proposed scheme using the real-time deployment of sensors and microprocessors (Raspberry-Pi).   

\subsection{Related work}
 \noindent In 1981, Lamport introduced the first RUA scheme based on the hash chain and with the password table at the server-side \cite{19Lamport}. By considering the limitations of password table-based schemes, in 1993, Chang et al. \cite{6Chang} introduced the first Smart Card (SC) based authentication scheme. In the SC-based authentication schemes, the user keeps an SC generated by the service provider as another security feature. Following this work, many other researchers proposed an RUA scheme for the client-server model used on the internet.
 
 In 2009, Das et al. set forth the first two-factor authentication scheme for the wireless sensor network \cite{9das}. In 2010, Khan et al. \cite{14islam} performed cryptanalysis on Das et al.'s scheme and successfully highlighted several vulnerabilities in their system. They underlined that Das et al. technique  is vulnerable to a range of threats, involving node bypassing, a lack of reciprocal authentication, and the likelihood of an insider attack.  Between 2010 to 2013, many authors proposed an authentication scheme, but the continuous fix-brake channel of the authentication brightens up new vulnerabilities and challenges for the two-factor authentication. 
 
 In 2013, Xue et al. \cite{31Xue} put forward a temporal credential-based mutual authentication scheme using a password for the WSN. In 2015, Jiang et al. \cite{15jiang} highlighted vulnerabilities like identity guessing, privilege insider attack, and stolen SC attack inside \cite{31Xue} and also proposed a new RUA scheme. In 2016, Amin et al. \cite{2amin} came up with a lightweight mutual authentication scheme for the WSN architecture. However, in 2017, Wu et al. \cite{29Wu} derived vulnerabilities like sensor capture attack, user forgery attack, gateway forgery attack, and user tracing attack security loopholes in \cite{2amin}. 
 
 In 2017, Jiang et al. \cite{15jiang} proposed a new three-factor scheme after analysis of the scheme proposed by Amin et al. They identified several security loopholes like known session-specific temporary attacks and tracking attacks, side-channel attacks, and impersonation attacks in \cite{2amin}. In 2017, Chen et al. \cite{8Chen} proposed an authentication scheme for the IoT environment and claimed that it is secured from all the well-known attacks, but recently in 2019, Patel et al. \cite{24Patel} provided cryptanalysis for Chen et al.'s scheme and proved that their scheme is not secure enough against attacks such as sensor device anonymity and gateway device bypassing attack. Recently, in 2020, Patel et al. also proposed a lightweight authentication scheme for the same network model \cite{25Patel}. They also highlighted that their proposed scheme is secure and lightweight for the U-GW network model. The proposed network model can be applied for UAV communications \cite{34Wang}, \cite{36Iwendi} also where users securely communicate with UAV devices. 
 
 We limit the related work discussion due to the restricted manuscript size. We suggest readers of this manuscript to refer other ECC based RUA schemes for further references \cite{30Xiong}, \cite{7Chaturvedi}, \cite{20Li}, \cite{25Patel}. 

\subsection{Road map of the paper} The remaining portion of the paper is organized as follows: \noindent The section \ref{sec:2preliminaries} contains the essential preliminaries that are employed throughout this work. We propose an authentication system in section \ref{sec:proposedscheme}. The section \ref{sec:securityanalysis} encompasses formal and informal security analysis which is followed by comparative analysis of the proposed system with existing schemes. The section \ref{sec:implement} provides implementation details. Finally, section \ref{sec:conclusion} summarizes this work by addressing some future aspects of the proposed work. 
\section{Preliminaries}
\label{sec:2preliminaries}
\noindent This section discusses preliminaries such as the Network model, Elliptic curve cryptography (ECC) encryption/decryption, and the threat model. We request readers to follow \cite{25Patel} for the basics of hash function and ECC.  
\subsection{\textbf{Network model}}
\label{sec:networkmodel}
\noindent IoT has started its journey with RFID (Radio-frequency identification) technology. RFID gets attraction due to its properties like being free from the line of sight and acceptable range. Thus, RFID tags on things can provide necessary information about a thing, and internet-connected RFID readers can help monitor and collect data. The IoT-based network models vary from application to application. In this paper, we consider the IoT-based U-GW network model \cite{25Patel}. We can use a smart home network as an example of the U-GW network model, where the user accesses data of all inner deployed sensing devices using the local gateway device. This model is also applicable to other IoT applications (such as smart hospital management), where every registered user receives data from the underlying sensing devices through the gateway device. 
\begin{figure}[H]
    \centering
    \includegraphics[width=\columnwidth]{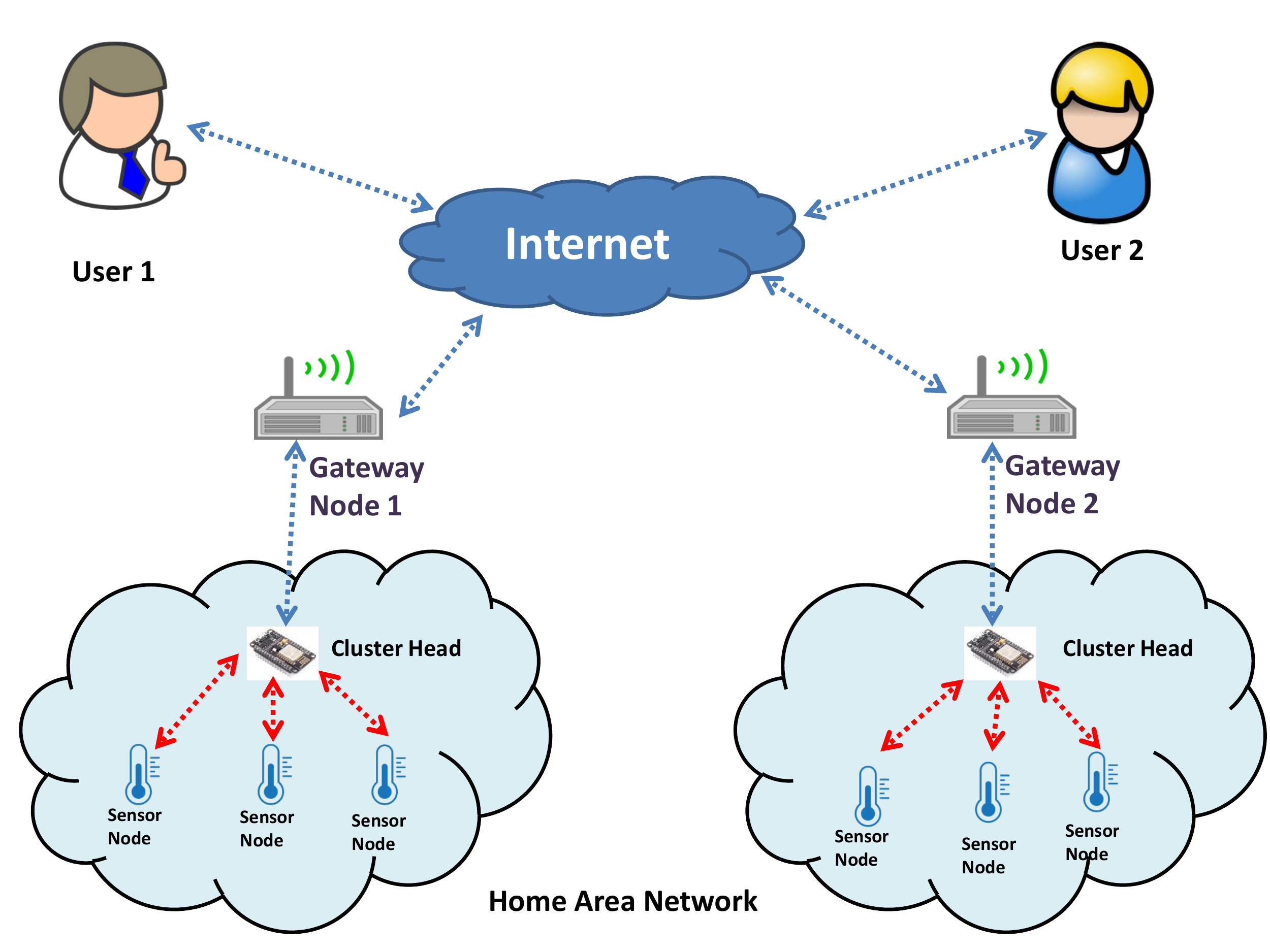}
    \caption{Generic IoT network model}
    \label{fig:1}
\end{figure}

Figure \ref{fig:1}. presents system model for the proposed approach. The user establishes a secure session key  only with the gateway device in this model. Over here, we make realistic assumptions that the communication between sensors and the home gateway is secured though we consider that communication as a locally secured network \cite{25Patel}. This assumption is valid because this model exclusively focuses only on the home network, and there is no foreign agent involved in it \cite{25Patel}. The gateway node forwards all received data to the cloud and provides live sensor data to the user. 

In this network model, all application users register with the gateway device, and the gateway device issues the requisite credentials based on their requirements for data access. Subsequently, whenever the IoT application user wants to access the data, they forward a request to the gateway node for the temporary lightweight session establishment. 
\subsection{\textbf{Notations and symbols}}
\noindent Table \ref{notations} presents notations used for articulation of the proposed scheme.
\begin{table}
   \centering
   \caption{\textbf{Notations and symbols}}
    \begin{tabular}{p{2cm} p{5cm}} \hline
    \textbf{Symbols}&\textbf{Description}  \\ \hline
         $UID$ & \emph{User Identity} \\
         $||$ & \emph{String concatenation} \\
         $UPW$ & \emph{User Password} \\
         $\oplus$ & \emph{Ex-OR Operation} \\
         $GW$ & \emph{Gateway} \\ 
         $S_k$ & \emph{Session key} \\ 
         $U$& \emph{User} \\ 
         $UR_i$ & \emph{User random number} \\ 
         $GWR_{j}$ \& $n_{gw}$ & \emph{Gateway random number} \\ 
         $Hash(.)$ & \emph{One way Hash function} \\
         $T_k$ & \emph{Time stamp} \\ 
         $K_{gw}$ & \emph{Gateway computed key} \\ 
         $K_{u}$ & \emph{User computed key} \\   
         $\mathcal{A}_d$ & \emph{Adversary} \\ 
         $\rightarrow$& \emph{Insecure communication} \\ 
         $\Rightarrow$ & \emph{Secured communication} \\ \hline 
    \end{tabular}
    
    \label{notations}
\end{table}
\subsection{\textbf{ECC encryption/decryption}}
\noindent Device A and Device B uses key derivation function (KBKDF) that provides every time same output if user provides same input every time. The encryption/decryption operation over elliptic curve occurs as follow:

\noindent \textbf{\textit{Device A performs}}
\begin{enumerate}[1.]
    \item Device A randomly generate \textit{$r_a$} where $r_a$ $\in$   \{1,2,....,n-1\} and is a private key for device A. 
    \item Device A computes \textit{$Pub_A$ = $r_a$ * G}. Here $Pub_A$ serves as a public key. For \textit{E/$F_P$}, \textit{G} is known as a public generator point shared between devices.
    \item Device A publishes \textit{$Pub_A$}. 
\end{enumerate}
\noindent \textbf{\textit{Device B performs}}
\begin{enumerate}[1.]
    \item Device B selects $r_b$ where $r_b$ $\in$ \{1,2,....,n-1\} and is a private key for device B.
    \item Device B computes public key \textit{$Pub_B$ = $r_b$ * G} \item Device B computes \textit{$r_b$ * $PUB_A$  =  $r_b$ *($r_a$ * G) = $r_a$ * ($r_b$ * G)}.
    \item Device B inputs \textit{$r_b$ * $PUB_A$} and generates symmetric key \textit{$K$} using the KBKDF. Device B encrypts message M as \textit{$M_E$ = $Enc_K(M)$}. Device B sends $M_E$ message to Device A along with the value \textit{$Pub_B$} through public channel.
\end{enumerate}
\noindent \textbf{\textit{Device A performs}}
\begin{enumerate}[1.]
    \item Device A inputs \textit{$r_a$ * ($Pub_B$)} into the same KBKDF and generates symmetric key for the decryption as \textit{$K$}. Device A decrypts $M_E$ as \textit{M = $Dec_K(M_E)$}.
\end{enumerate}
\subsection{\textbf{Threat Model}}
\noindent We follow the following threat model considered from \cite{10Dolev} which derives certain capabilities for the adversary $\mathcal{A}_d$. The polynomial time $\mathcal{A}_d$ can: 
\begin{enumerate}[1]
   \item compute valid pair of $identity * password$ offline in the polynomial time.
    \item extract information from the user's smart card \cite{22Messerges}. 
    \item fully access communication channel between U-GW.
    \item obtain the previously calculated session key.
    \item get secrets of gateway node during system failure situations \cite{18kocher}. 
    \item get physical access to user device and can extract saved information through power analysis.
    \item Adversary $\mathcal{A}_d$ is unable to access the user password or SC simultaneously. The simultaneous access of both the password and SC leads to impersonation attack \cite{25Patel}.
\end{enumerate}

\section{Proposed Scheme}\label{proposed}
\noindent This section discusses the proposed ECC-based two-factor authentication scheme for the U-GW paradigm. The proposed scheme is divided into four steps: (1) Setup phase, (2) User registration phase, (3) Login and key Generation phase, and (4) Password update phase. We assume that there are \textit{n} users and a single gateway device for the proposed work. Table \ref{login} bright ups tabular presentation for the login and key generation phase. 
\label{sec:proposedscheme}
\subsection{\textbf{Setup phase}}
\noindent During the setup phase, the user device ($U_i$) and the gateway device ($GW$) setup computation environment. For the proposed scheme, the setup phase performs as follow: 

\noindent \textbf{\textit{UF1:}}
\begin{enumerate}[Step 1:]
    \item Both $U_i$ and $GW$ devices agree on \textit{$E/F_p$} (represents curve defined over finite field F using prime number p), the generator point \textit{G(X,Y)} and the KBKDF.
    \item $U_i$ produces random number \textit{$UR_i$} in such a way \textit{$UR_i \in Z_p^{*}$}, Here p is a large prime number and the range of \textit{$Z_p^{*}$} is \textit{(0,1,2....,p-1)}.
    \item $U_i$ calculates \textit{$D_u$ = $UR_i*G$}.  
    \item $U_i$ sends \textit{$D_u$} to the $GW$ in offline manner. 
\end{enumerate}
\noindent \textbf{\textit{GWF1:}}
\begin{enumerate}[Step 1:]
    \item $GW$ generates random number \textit{$GWR_{j}$} where \textit{$GWR_{j} \in Z_p^{*}$}.
    \item $GW$ calculates \textit{$S_i$ = $GWR_{j}$*G}. 
    \item $GW$ transmits \textit{$S_i$} to the $U_i$.
\end{enumerate}    
\noindent \textbf{\textit{UF2:}}
\begin{enumerate}[Step 1:]
\item $U_i$ generates symmetric key \textit{$K_{u}$} using \textit{$UR_i$,$S_i$} by KBKDF. This is not the session key but this is a key that is used for encryption/decryption by both $U_i$ and $GW$ whenever they establish session.
\end{enumerate} 
\noindent \textbf{\textit{GWF2:}}
\begin{enumerate}[Step 1:]    
\item $GW$ generates symmetric key \textit{$K_{gw}$} using \textit{$GWR_{j}$,$D$} by KBKDF. 
\item $GW$ discards $D$.
\end{enumerate}
Thus, after completion of setup phase, $U_i$ device have \{$UR_i$, $D_u$, $S_i$\} in its confidential memory while $K_{gw}$ in the secret memory and the $GW$ device have \{$GWR_{j}$, $S_i$\} in its confidential memory while $K_{gw}$ in its secret memory. The data stored in confidential memory is read-only and can be read by anyone using power analysis or reverse engineering. But, the data stored in secret memory can be read by the device itself only.

\subsection{\textbf{User Registration Phase}}
\noindent In this phase, the $U_i$ registers to the $GW$ by using identity $UID$ and password $UPW$. The gateway device generates an SC (SC) with numerous parameters. These parameters are used by $U_i$ during the login and key generation phase. 

\noindent \textbf{\textit{UF1:}}
    \begin{enumerate}[Step 1:]
            \item $U_i$ chooses unique identity $UID$ and secure password $UPW$.
            \item $U_i$ calculates $h_i$ = Hash($D_u || UPW || UID$). 
            \item $U_i$ generates message $M_1$ = \{$UID$,$h_i$\} and sends $M_1$ to $GW$ over secure channel. We created a Transport Layer Security (TLS) channel over the MQTT protocol for data transmission across a secure means in our implementation. 
    \end{enumerate}
\noindent \textbf{\textit{GWF1:}} 
    \begin{enumerate}[Step 1:]
        \item $GW$ receives message $M_1$ and retrieves user identity $UID$. $GW$ verifies $UID$ in the database and checks its availability. If similar $UID$ is used by another user then $GW$ transmits message such as "Identity not available", else ($UID$ is available for user $U_i$) $GW$ continues. 
        \item $GW$ computes $X_i$ = Hash($UID \oplus h_i$), \\ $MID$ = $Enc_{K_{gw}}$($UID$), $T_1$ = Hash($S_i || K_S$), $O_i$ = $T_1 \oplus h_i$.
        \item $GW$ frames smart card \textit{SC}, $SC$ = \{$O_i$,MID,$X_i$,Hash(.)\}. 
        \item $GW$ sends \textit{SC} to the user through secure channel. 
    \end{enumerate}
\noindent \textbf{\textit{UF2:}}
    \begin{enumerate}[Step 1:]
        \item User $U_i$ acquires \textit{SC} from the $GW$ device and keep $D_u$ after encrypting it using $TK_{gw}$ = \{$UPW$*P\} and generates \textit{Z} = $Enc_{TK_{gw}}$($D_u$). $U_i$ stores \textit{Z} into secret memory of the \textit{SC}. Thus final \textit{SC} with user have \{$O_i$,$MID$,$X_i$,$Hash()$,\textit{Z}\} parameters. Except \textit{Z}, remaining parameters will be in confidential memory of the \textit{SC}.
    \end{enumerate}
\subsection{\textbf{Login and key generation phase}}
\noindent The user will supply their identify, password, and SC to the card reader during this phase. The card reader will authenticate the user's identity and password before communicating with the Gateway to establish the session key.  \\ 
\noindent \textbf{\textit{UF1:}}
  \begin{enumerate}[Step 1:]
      \item $U_i$ inserts $UID$, $UPW$ and \textit{SC} in to SC reader (\textit{SCR}). \textit{SC} computes $Dec_{UPW*P}(Z)$ and extracts $D$ from its secret memory.
      \item calculates $h_i$ = Hash($D_u || UPW || UID$), $X_i^{*}$ = Hash($UID||h_i$), and verify $X_i^{*} \stackrel{?}{=} X_i$, T = $O_i \oplus h_i$, $L_i$ = Hash($D_u || UID$), PID = T $\oplus$ Hash($UID || L_i || T_k$), Z = $Enc_{K_{gw}}$($D_u$).
      \item $U_i$ frames \textit{M = \{MID,Z,$T_{ki}$,PID\}} and sends it to the gateway device. Over here $T_{ki}$ is the current time stamp of User.
  \end{enumerate}
\noindent\textbf{\textit{GWF1:}}
  \begin{enumerate}[Step 1:]
      \item $GW$ receives \textit{message M} and computes own timestamp $T_k^{*}$ , Verify timestamp $\Delta T \stackrel{?}{\leq} T_k^{*}$ - $T_k$, 
      \item $GW$ calculates $Dec_{K_{gw}}(MID)$ to receive $UID$.
      \item $GW$ extracts \textit{$D_u$} by calculating $Dec_{K_{gw}}(Z)$.
      \item $GW$ computes T = Hash($GWR_{j}||K_{gw}$), $N_i$ = Hash($D_u || UID$), $PID^{*}$ = T $\oplus$ Hash($UID||N_i||T_{ki}$). 
      \item $GW$ verifies $PID \stackrel{?}{=} PID^{*}$. After successful verification only $GW$ continues.
      \item $GW$ generates $n_{gw}$ in order to $n_{gw} \in Z_p^{*}$ and computes NS = $Enc_{K_{gw}}(n_{gw})$. 
      \item $GW$ calculates $S_k$ = Hash($UID || T || n_{gw} * D_u$).
      \item $GW$ calculates key verifier $SQ_i$ = \\ Hash($S_k || n_{gw} || T || T_{k_{new}}$) and forwards $M_{new}$ = \{$SQ_i,NS,T_{k_{new}}$\} to user.
  \end{enumerate}
\noindent \textbf{\textit{UF2:}}
    \begin{enumerate}[Step 1:]
        \item $U_i$ recovers $n_{gw}$ = $Dec_{K_{u}}(NS)$.
        \item $U_i$ calculates $S_k^{*}$ = Hash($UID||T|| n_{gw} * D_u$).
        \item $U_i$ computes $SQ_i^{*}$ = Hash($S_k^{*}||n_{gw}||T||T_{k_{new}}$).
        \item $U_i$ verifies $SQ_i \stackrel{?}{=} SQ_i^{*}$. If verification is successful then he/she considers $S_k^{*}$ as a new current session key.  
    \end{enumerate}
\begin{table*}
 \centering
  \caption{\textbf{Proposed Scheme : Login and Authentication phases}}
    \begin{tabular}{|p{8cm}|p{7cm}|} \hline
    \textbf{{User ($U_i$)}} & \textbf{{Gateway (GW)}}  \\ \hline
    \textbf{{The Login \& Authentication Phase:}} & \\
    {Enters \textit{SC}, provides $UID$, $UPW$,} & \\ 
    {Computes $h_i$ = Hash($UID||UPW||Dec_{UPW}(Z)$),$X_i^{*}$ = Hash($UID \oplus h_i$), $X_i^{*} \stackrel{?}{=} X_i$,}  & \\ 
    {Computes $T$ = $O_i \oplus h_i$, $L_i$ = Hash($D || UID$),} & \\ 
    {Computes $PID$ = T $\oplus$ Hash($UID || L_i || T_k$), Z = $Enc_{K_{u}}(D)$,} & \\
    {$\xrightarrow{\{PID, T_{ki}, MID, Z\}}$} &  \\ 
     & {Gets current time $T_k^{*}$} \\  
     & {Verifies $\Delta T \stackrel{?}{\leq} T_k^{*}$ - $T_k$,$Dec_{K_{gw}}$(MID),$Dec_{K_{gw}}$(Z)}\\ 
     & {Computes $N_i$ = Hash($ D_u || UID$),T = Hash($GWR_{j} || K_{gw}$),$PID^{*}$ = T $\oplus$ Hash($UID || N_i || T_k$),} \\
     & {Verifies $PID^{*} \stackrel{?}{=} PID$, Generates $n_{gw} \in Z_p^{*}$,}\\
     & {Computes NS = $Enc_{K_{gw}}$($n_{gw}$), $S_k$ = Hash($UID || T || n_{gw}* D_u$), $SQ_i$ = Hash($S_k || n_{gw} || T || T_{k_{new}}$)} \\ 
    {Gets $T_{k_{new}}^{*}$, Verifies $\Delta T_{new} \stackrel{?}{\leq} T_{k_{new}}^{*} - T_{k_{new}}$,} & {$\xleftarrow{NS, SQ_i,T_{k_{new}}}$}\\
    {Computes $n_{gw}$ = $Dec_{K_{u}}(NS)$} &  \\
    {Computes $S_k^{*}$ = Hash($UID || T || n_{gw}* D_u$),} & \\
    {Computes $SQ_i^{*}$ = Hash($S_k^{*} || n_{gw} || T || T_{k_{new}}$),} & \\ 
    {Verifies $SQ_i^{*} \stackrel{?}{=} SQ_i$,} & \\
    {Session key is $S_k$ = $S_k^{*}$} & \\ \hline
    \end{tabular}
    \label{login}
\end{table*}

\subsection{\textbf{Password update phase}}
\noindent Using this phase, $U_i$ updates his/her identity $UID$ and password $UPW$ through SCR device. 

\noindent \textbf{\textit{UF1:}}
\begin{enumerate}[Step 1:]
    \item $U_i$ provides \textit{SC} to SCR and selects password update option.
    \item $U_i$ provides $UID$, $UPW$, $UPW^{new}$ 
\end{enumerate}
\noindent \textbf{\textit{SCRF1:}}
\begin{enumerate}[Step 1:]
    \item SCR validates \{$UID$, $UPW$\}. SCR computes $h_i$ = Hash($UID|| UPW ||D$), $X_i^{*}$ = Hash($UID \oplus h_i$) and checks $X_i^{*} \stackrel{?}{=} X_i$. 
    \item Computes $h_i^{new}$ = Hash($UID || UPW^{new} || D_u$).
    \item Computes $O_i$ = $h_i \oplus h_i^{new} \oplus O_i$, $X_i^{new}$ =  Hash($UID \oplus h_i^{new}$), $Z_{new}$ = $Enc_{UPW^{new}}(D_u)$.
    \item Replaces \{$X_i$,$O_i$,\textit{Z}\} by \{$X_i^{new}$,$O_i^{new}$,$Z_{new}$\}. 
\end{enumerate}

\section{Security Analysis}
\label{sec:securityanalysis}
\noindent In this section, we put forward an informal security analysis using the Dolev-Yao channel \cite{10Dolev} and a formal security analysis using the AVISPA and the BAN Logic.
\subsection{\textbf{Informal security analysis}}
\noindent We present an informal security analysis for the proposed authentication mechanism employing the Doleve-Yao threat model in this subsection.  \cite{10Dolev}. Table \ref{Tab:Seccomp} provides a security comparison for the proposed scheme with the other existing schemes. 
\label{sec:informalsecurityanalysis}
\begin{enumerate}[F1.]
\item {\textbf{Password guessing attack:}}\\ 
\noindent In this attack, adversary $\mathcal{A}_d$ performs offline and online password assumptions and verifies numerous passwords. $\mathcal{A}_d$ uses a famous bruit force dictionary method to get success in this attack. In our scheme, an $\mathcal{A}_d$ can not achieve success in password guessing attacks due to hashing of Hash($UID || UPW || D_u$). Even though adversary may guess the correct $UID$ and $UPW$; then also, he/she will not receive value of \textit{$D_u$}. After guessing the correct \{$UID$, $UPW$\} pair, an adversary can not compute the final session key because of the unavailability of the random parameter $ns$. As a result, the proposed scheme is resistant to the password guessing attack. 
\noindent
\item {\textbf{Message Replay attack:}}\\
\noindent In this attack, the polynomial adversary $\mathcal{A}_d$ captures communication between $U_i$ and $GW$. These messages are used for spoofing and impersonation type activities. We use three random parameters that provide immunity from the replay activities of the attacker. The first parameter is a timestamp ($T_k$), the second one is a random variable $UR_i$, and the last parameter is random nonce $N_i$. With the help of these parameters, the receiver can quickly validate the freshness of received messages. As a result, the proposed scheme is protected against the message replay attack.
\noindent
\item {\textbf{User anonymity:}}\\
\noindent If adversary $\mathcal{A} d$ acquires the user's identity, we can state that the user anonymity attack is successful. In our scheme, we use hash function to protect identity and password at user side ($h_i$ = Hash($UID||UPW||Dec_{UPW}$)) as well as at gateway side, we validate inside hash computations. Thus, neither entity shares $UID$ over an insecure channel in the proposed scheme.
\noindent
\item {\textbf{Perfect forward secrecy attack:}}\\
\noindent In this threat, the availability of a gateway secret key $K_{gw}$ does not lead an adversary towards successful session key computations. The proposed scheme uses a random number ($n_{gw}$) at the gateway side for the session key computation; thus, even though the adversary gets the gateway secret key $K_{gw}$, they do not succeed in the previous session key as well as session key for future communications.    
\noindent
\item {\textbf{Stolen SC attack:}}\\
\noindent In the stolen SC attack, $\mathcal{A}_d$ receives SC and performs power analysis to extract data. After extracting data, $\mathcal{A}_d$ require password to receive $Z$. Thus, the extraction of SC data does not provide the right direction to an adversary for the session key computations. As a result, the proposed scheme is resistant to stolen SC threats. 
\noindent
\item {\textbf{Privileged insider attack:}}\\
\noindent In this attack, $\mathcal{A}_d$ is an insider to the gateway device and can see the received messages. We do not relieve $UID$ anywhere in the plaintext during computation in the proposed scheme. After receiving $PID$ from the user, the gateway performs verification for the $PID^{*}$ = $PID$, and it is secured using hash. Thus, $\mathcal{A}_d$ at the gateway device neither receive identity nor password. For key computation, $\mathcal{A}_d$ needs \{$UID$,$T$,$D_u$\} and it is near to impossible for an $\mathcal{A}_d$ to re-frame the same pair. As a result, the proposed scheme is impervious to privileged insider attacks.   

\item {\textbf{Mutual authentication:}}\\
\noindent The proposed scheme achieves mutual authentication. The gateway $GW$ calculates $PID^{*}$ and validate it with the $PID^{*}$. The $GW$ also validates ID and PW of the $U_i$. The $U_i$ authenticates $GW$ by verifying $Q_i^{*} \stackrel{?}{=} Q_i$. The $Q_i$ is calculated with the help of received key. As a result, we may assert that the suggested system meets the mutual authentication property. 

\item {\textbf{Man-In-The-Middle (MITM) attack:}}\\
\noindent In this attack, $\mathcal{A}_d$ receives messages communicated between the $U_i$ and the $GW$. In proposed scheme, even $\mathcal{A}_d$ receives \{$PID$, $MID$, $Z$\} then also he/she will not get success in order to read inside data due to hashing and ECC encryption. Using \{$NS$ (Computed parameter using $n_{gw}$) , $SQ_i$\}, $\mathcal{A}_d$ could not calculate the $S_k$ because of inadequate information about \{$n_{gw}$, $UID$, $D$\}. As a result, the suggested technique is protected against an MITM attack. 
\noindent
\item {\textbf{User impersonation attack:}}\\
\noindent To impersonate as a legitimate user, an adversary $\mathcal{A}_d$ captures SC parameters such as \{$O_i$, $S_i$, $A_i$, Z\}. Using these parameters, an  $\mathcal{A}_d$ can not generate correct login requests \{$PID$, $T_{ki}$, $MID$, Z\}. As a result, an adversary's imitation attempts are unsuccessful. Hence, the proposed scheme is secured against the user impersonation attacks. 

\noindent
\begin{table*}
    \centering
    \caption{\textbf{Security comparison}}
     \begin{tabular}{|p{1cm} p{5cm} p{1cm} p{1cm} p{1cm} p{1cm} p{1cm} p{1.5cm}|}  \hline
    & Security Parameter & \cite{14islam} & \cite{33Zhang} & \cite{23Odelu} & \cite{21Liu} & \cite{26Qiu} &  Proposed \\  \hline
     F1 & Offline Password guessing & $\times$ & $\times$ & \checkmark & \checkmark & $\times$ & \checkmark  \\ 
     F2 & Replay & \checkmark & \checkmark & $\times$ & $\times$ & \checkmark & \checkmark  \\ 
     F3 & User anonymity & \checkmark & \checkmark & \checkmark & \checkmark & $\times$ & \checkmark  \\ 
     F4 & Perfect forward secrecy & \checkmark & \checkmark & \checkmark & \checkmark & \checkmark & \checkmark  \\ 
    F5 & Stolen smart card & $\times$ & $\times$ & \checkmark & \checkmark & \checkmark & \checkmark  \\ 
    F6 & Privilege insider & \checkmark & \checkmark & \checkmark & \checkmark & \checkmark & \checkmark  \\ 
    F7 & Mutual authentication & \checkmark & \checkmark & \checkmark & \checkmark & \checkmark & \checkmark  \\
    F8 & Man-in-the-Middle & $\times$ & \checkmark & \checkmark & \checkmark & \checkmark & \checkmark  \\
    F9 & User impersonation & $\times$ & \checkmark & \checkmark & $\times$ & $\times$ & \checkmark  \\ 
    F10 & Gateway impersonation & \checkmark & \checkmark & \checkmark & $\times$ & $\times$ & \checkmark  \\ 
   F11 & Denial of Service & \checkmark & \checkmark & \checkmark & \checkmark & \checkmark & \checkmark  \\ \hline
     \end{tabular}
     \label{Tab:Seccomp}
\end{table*}
\item {\textbf{Gateway impersonation attack:}}\\
\noindent To impersonate as a legitimate gateway, an adversary $\mathcal{A}_d$ captures the public message and public parameters such as the gateway device. Using these parameters, an  $\mathcal{A}_d$ can not generate correct reply $M_{new}$ = \{$SQ_i,NS,T_{k_{new}}$\} for the user due to inadequacy of parameter $K_{gw}$. Thus, an adversary does not get success in user impersonations. Similarly, an adversary $\mathcal{A}_d$ can not decrypt the message $MID$ and $Z$ due to the nonavailability of gateway master secrets. As a result, the proposed scheme is unlikely to be affected by a gateway impersonation attack. 

\item {\textbf{Denial of service attack:}}\\
Using this attack, the polynomial adversary $\mathcal{A}_d$ tries to stop $U_i$ and $GW$ from key generation through either flooding or any other means. An adversary $\mathcal{A}_d$ gives rise to redundant requests for either device and generates too much delay in the key generation process. In the real-time scenario, it is nearly impossible to achieve full proof protection against DoS-based attacks; still, we tried to protect our scheme using the time-stamp and random numbers. As a result, it prevents an attacker $\mathcal{A}_d$  from gaining complete control of the system. 
\end{enumerate}
\subsection{\textbf{Mutual authentication using BAN logic}}
\label{sec:banlogic}
\noindent In this subsection, we show that our scheme achieves mutual authentication property using this BAN logic tool. The BAN logic generates trust between the principles (communicating parties). It focuses on the proposed scheme's coherence and feasibility. It works on proper formulations of postulates, inference rules, and assumptions with realistic goals. 
\subsection{\textbf{Postulates}}
\noindent Figure \ref{fig:Post} shows basic postulates used by the BAN Logic. Over here \textit{P} and \textit{Q} are the communicating principals, \textit{M1} and \textit{M2} are the communicated messages. The key \textit{KS} is used for the encryption/decryption operations and $k_s$ is the shared secret.
\begin{figure*}
    \centering
    \includegraphics[width=0.75\textwidth]{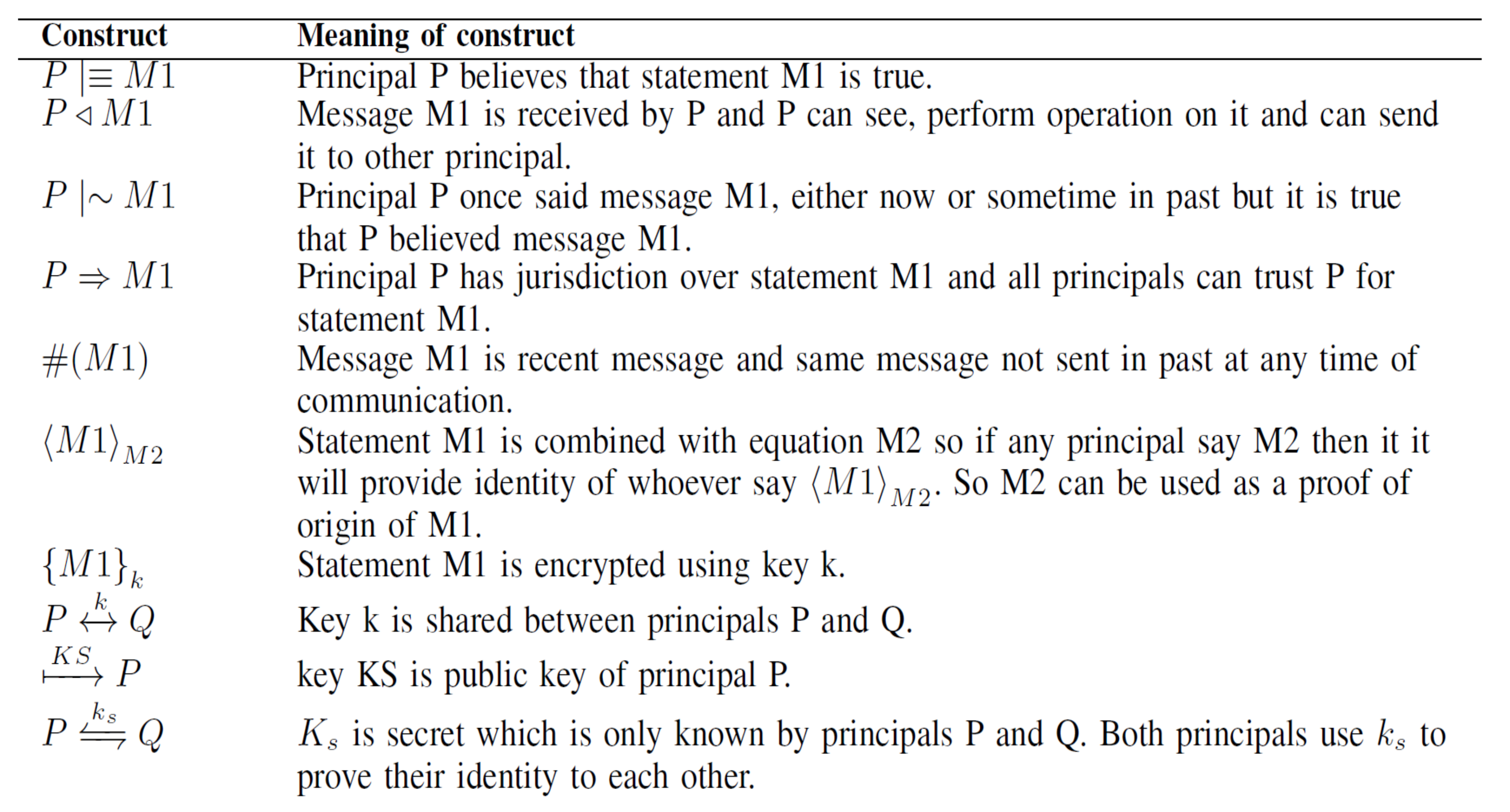}
    \caption{Postulates}
    \label{fig:Post}
\end{figure*}
\subsection{\textbf{Inference rules}}
\noindent Inference rules are derivations that are derived from postulates by the BAN logic. The inference rules prove that the proposed authentication scheme satisfies the mutual authentication following. Figure \ref{fig:infer} shows basic inference rules generated for mutual authentication proof.
\begin{figure*}
    \centering
     \includegraphics[width=0.75\textwidth]{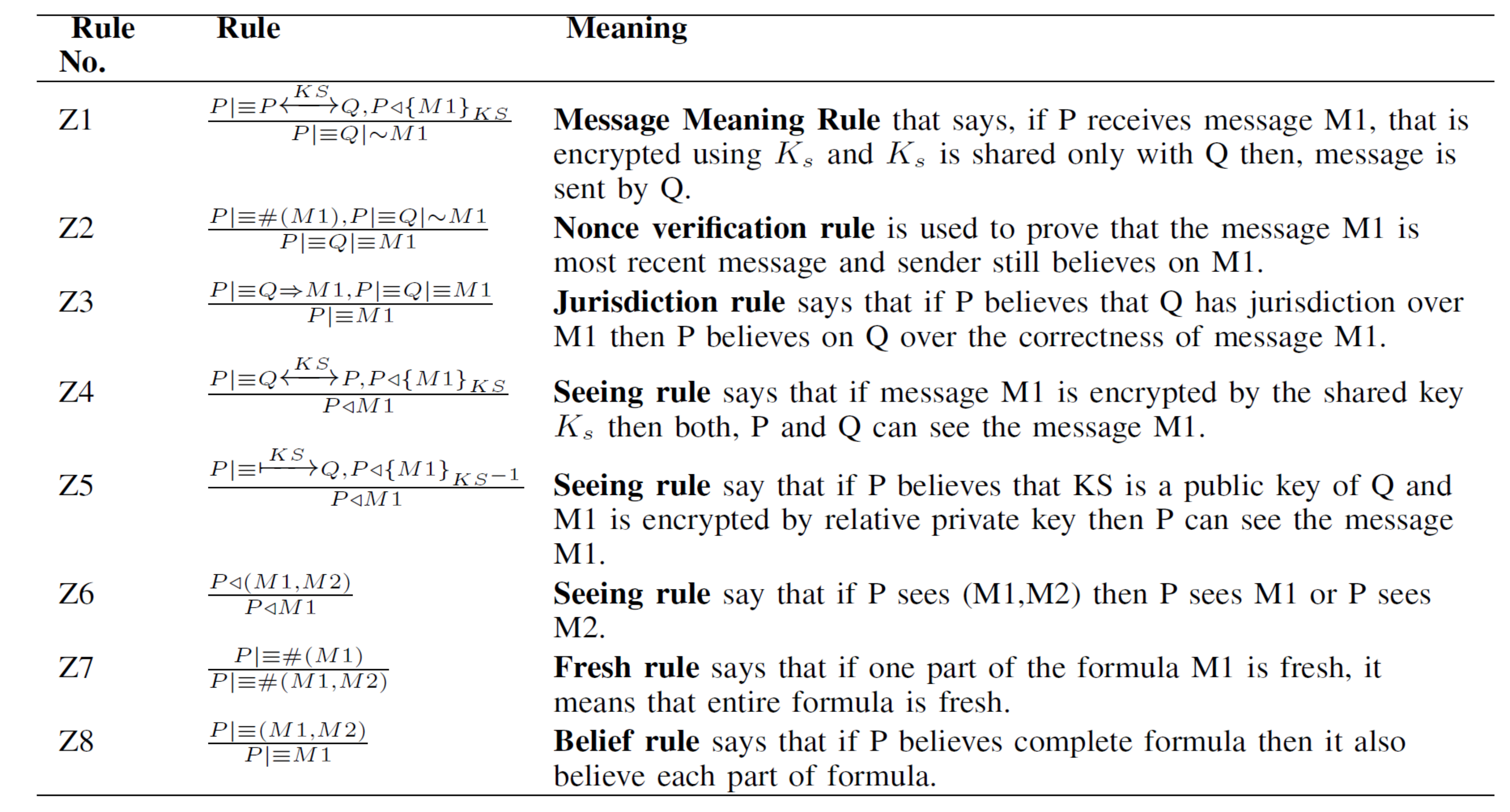}
    \caption{Inference Rules}
    \label{fig:infer}
\end{figure*}
\subsection{\textbf{Assumptions}}
\noindent BAN Logic works based on the following assumptions:
\begin{itemize}
    \item There are some shared secrets.
    \item Principals can compute fresh nonces.
    \item Each principals believes on each other.
    \item Each principal can recognize his/her messages.
    \item If principal \textit{P} believes that \textit{KS} is his public key, then \textit{P} must know corresponding private key $KS^{-1}$. 
\end{itemize}
\subsection{\textbf{Goals}}
\noindent For any U-GW model, the followings are the primary goals or conclusions those must be achieved during the authentication. Hence, using above said inference rules, assumptions and postulates, both user ($U_i$) and the gateway ($GW$) have to achieve following goals:
\begin{table}
    \centering
    \caption{Goals}
    \begin{tabular}{p{1cm}p{5cm}} \hline
     Goal No. & Goal  \\ \hline
        1 & User ${U_i} \believes C \sharekey{S_k} GW$ \\ 
        2 & Gateway $GW \believes S \sharekey{S_k} {U_i}$ \\ \hline
    \end{tabular}
    \label{goals}
\end{table}
\begin{theorem}
The proposed satisfies mutual authentication property among $U_i$ and $GW$.
\end{theorem}
\begin{proof}
We rewrite the messages communicated in login and authentication phase of the proposed scheme in the generic form as below:

\textit{\textbf{MS 1:}} $U_i  \rightarrow GW$: (Enc(ID), Hash($S_i||K_{gw}$)$\oplus$ \\
Hash($D||UPW||UID$) $\oplus$ Hash($D||UID||UPW$) $\oplus$ \\ Hash($UID$ $|| Hash(D ||UID)$ $||T_{k_i}$),\textit{Enc(D)}, $Ts_i$)

\textit{\textbf{MS 2:}} $GW \rightarrow U_i$: (\textit{Enc($n_{gw}$)}, \\ Hash(Hash($UID$ $||Hash(GWR_{j}||K_{gw})$ $||ns*D$) \\ $||Hash(GWR_{j}||K_{gw})$ $||T_{k_{new}}$), $T_{k_{new}}$)

\noindent \textbf{Idealized form:} We can rephrase MS1 and MS2 in idealized form as follows:

\textit{\textbf{MS 1:}} $U_i  \rightarrow GW$: $\langle<$ {\textit{(ID)}, ($S_i||K_{gw}$),\\
($D||UPW||UID$), ($D||UID||UPW$) , ($UID$ $|| \\ Hash(D ||UID)$ $||T_{ki}$),\textit{(D)}, $Ts_i$ $> \rangle$}$_{U_i \sharekey{D,S_i} GW}$

\textit{\textbf{MS 2:}} $GW \rightarrow U_i$: $\langle<$ ($n_{gw}$), (($UID$ $||(GWR_{j}||K_{gw})$ $||ns*D$) $||(GWR_{j}||K_{gw})$ $||T_{k_{new}}$), $T_{k_{new}}$ $>\rangle$ $_{U_i \sharekey{D,S_i} GW}$

\noindent \textbf{Goal:} We define goals of the proposed scheme in idealized form as follow:

\textbf{GL1:}  $GW \believes U_i \sharekey{S_k} GW$ 

\textbf{GL1:}  $U_i \believes U_i \sharekey{S_k} GW$

\noindent Following assumptions are used to prove the mutual authentication:

\textit{\textbf{Y1.}} $U_i \believes \fresh{T_{ki}}$

\textit{\textbf{Y2.}} $GW \believes \fresh{T_{ki}}$

\textit{\textbf{Y3.}} $U_i \believes \fresh{T_{k_{new}}}$

\textit{\textbf{Y4.}} $GW \believes \fresh{T_{k_{new}}}$

\textit{\textbf{Y5.}} $U_i \believes (U_i \sharekey{S_i,D} GW)$

\textit{\textbf{Y6.}} $GW \believes (U_i \sharekey{S_i,D} GW)$

\textit{\textbf{Y7.}} $U_i \believes GW \oncesaid (U_i \sharekey{S_k} GW)$

\textit{\textbf{Y8.}} $GW \believes \fresh{n_{gw}}$

\textit{\textbf{Y9.}} $U_i \believes \fresh{n_{gw}}$

\noindent Mutual authentication among $U_i$ and $GW$ is achieved as given below:

\textit{\textbf{W1:}} From \textit{MS 1},

$GW \sees \langle<$ {(ID), ($S_i||K_{gw}$),
($D||UPW||UID$), \\ ($D||UID||UPW$) , ($UID$ $|| Hash(D ||UID)$ $||T_{ki}$),\\ (D), $Ts_i$ $> \rangle$}$_{U_i \sharekey{D,S_i} GW}$

\textit{\textbf{W2:}} Using \textit{W1}, \textit{Z1} and \textit{Y6}, we obtain,

$GW \believes U_i \oncesaid \langle<$ {(ID), ($S_i||K_{gw}$),
($D||UPW||UID$), ($D||UID||UPW$) , ($UID$ $|| Hash(D ||UID)$ $||T_{ki}$),(D), $Ts_i$ $> \rangle$}

\textit{\textbf{W3:}} Using \textit{W2}, \textit{Y2}, \textit{Z2}, we gain,

$GW \believes U_i \controls$ $\langle<$ {\textit{(ID)}, ($S_i||K_{gw}$),
($D||UPW||UID$), ($D||UID||UPW$) , ($UID$ $|| Hash(D ||UID)$ $||T_{ki}$),\textit{(D)}, $Ts_i$ $> \rangle$}  

\textit{\textbf{W4:}} Using \textit{Z7}, \textit{Z3}, we get,

$GW \believes  \fresh $ $\langle<$ {(ID), ($S_i||K_{gw}$),
($D||UPW||UID$), \\ ($D||UID||UPW$) , ($UID$ $|| Hash(D ||UID)$ $||T_{ki}$),(D), $Ts_i$ $> \rangle$}  

\textit{\textbf{W5:}} Using \textit{W3},\textit{W4}, \textit{Z8}, we achieve,

$GW \believes U_i \controls$ ($S_i||K_{gw}$),
($D||UPW||UID$), \\ ($D||UID||UPW$) , ($UID$ $|| Hash(D ||UID)$ $||T_{ki}$)

\textit{\textbf{W6:}} Using \textit{W1}, \textit{W4}, \textit{Z2}, \textit{Y8}, we obtain,

$GW \believes GW \sharekey{S_k} U_i$ \textbf{\textit{[Goal 1]}}

\textit{\textbf{W7:}} From \textit{message 2}, we receive,

$ U_i \sees$ $\langle<$ ($n_{gw}$), (($UID$ $||(GWR_{j}||K_{gw})$ $||ns*D$)$ \\ ||(GWR_{j}||K_{gw})$ $||T_{k_{new}}$), $T_{k_{new}}$ $>\rangle$ $_{U_i \sharekey{D,S_i} GW}$

\textit{\textbf{W8:}}  Using \textit{W7}, \textit{Z1} and \textit{Y6}, we receive,

$ U_i \believes GW \oncesaid $ $\langle<$ ($n_{gw}$), (($UID$ $||(GWR_{j}||K_{gw})$ $||ns*D$) $||(GWR_{j}||K_{gw})$ $||T_{k_{new}}$), $T_{k_{new}}$ $>\rangle$

\textit{\textbf{W9:}} Using \textit{W8}, \textit{Y2}, \textit{Z2}, we obtain,

$ U_i \believes GW \controls $ $\langle<$ ($n_{gw}$), (($UID$ $||(GWR_{j}||K_{gw})$ $||ns*D$) $||(GWR_{j}||K_{gw})$ $||T_{k_{new}}$), $T_{k_{new}}$ $>\rangle$

\textit{\textbf{W10:}} Using \textit{W9}, \textit{Z7},\textit{Z3}, we get,

$ U_i \believes \fresh$ $\langle<$ ($n_{gw}$), (($UID$ $||(GWR_{j}||K_{gw})$ $||ns*D$) $||(GWR_{j}||K_{gw})$ $||T_{k_{new}}$), $T_{k_{new}}$ $>\rangle$

\textit{\textbf{W11:}} Using \textit{W10}, \textit{W9}, \textit{Z8}, we achieve,

$ U_i \believes GW \controls $ $\langle<$ (($UID$ $||(GWR_{j}||K_{gw})$ $||ns*D$) $>\rangle$

\textit{\textbf{W12:}} Using \textit{W8}, \textit{W9}, \textit{Z2}, \textit{Y8}, we receive,

$U_i \believes U_i \sharekey{S_k} GW$ \textbf{\textit{[Goal 2]}}

Thus, \textit{Goal 1} and \textit{Goal 2} derives that proposed work satisfies mutual authentication among communicating entities. 
\end{proof}

\section{Formal security Simulation Using AVISPA}
\label{sec:avispa}
\noindent Formal security protocol simulation provides strong protocol verification against cryptographic and networking attacks. AVISPA is widely adopted formal security analysis tool. AVISPA is a one-button tool that verifies the security of protocol against attacks such as replay and MITM attacks. We utilized AVISPA for further security analysis of our protocol. Because of its unique simulation environment, performing simulation in AVISPA is a tremendous challenge. Figure \ref{avispa} shows flow if AVISPA simulation. \\
\begin{figure}
    \centering
     \includegraphics[width=\linewidth]{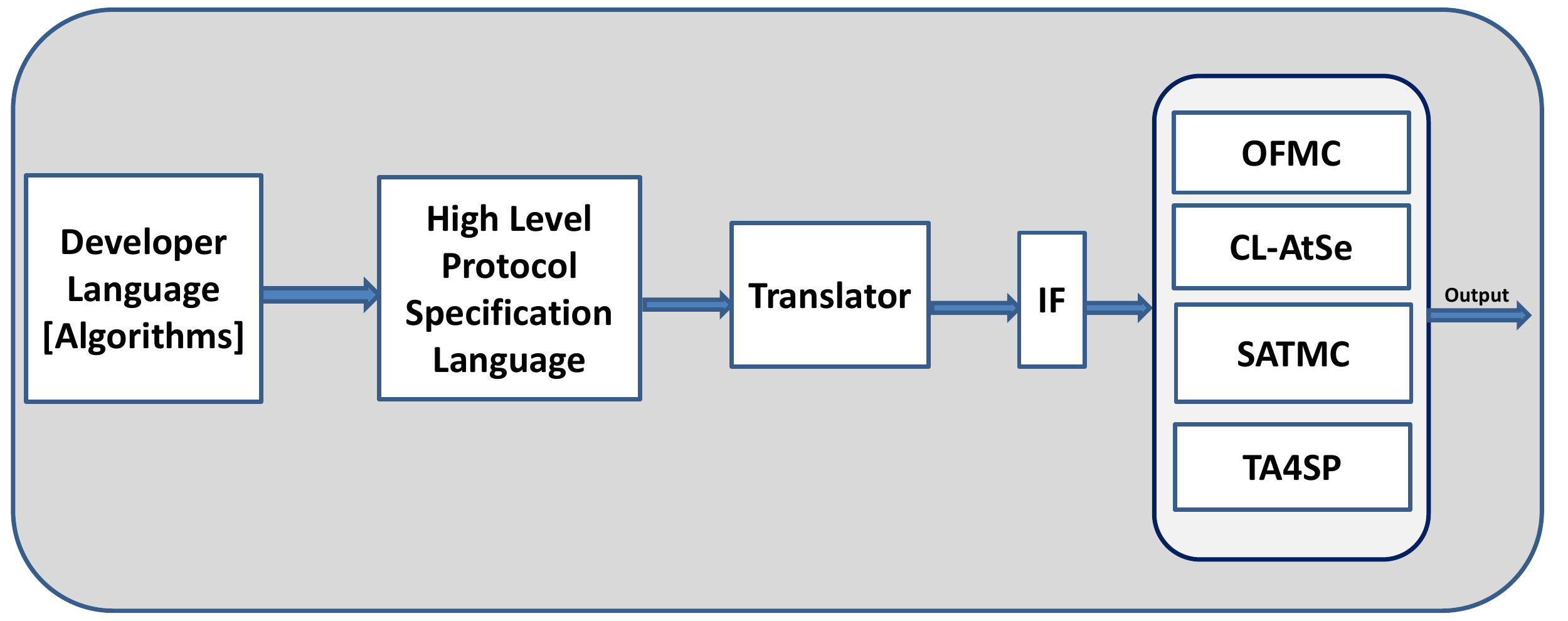}
    \caption{AVISPA Simulation Flow}
    \label{avispa}
\end{figure}
AVISPA conducts its simulations in a difficult-to-read language called High-Level Protocol Specification Language (HLPSL), which is less opaque.
HLPSL is a role-based language for describing proposed protocol and intruder behaviors. As a result, each proposed protocol in AVISPA must be rebuilt in HLPSL.
In HLPSL, roles represent principals, sessions, and the environment. There are three types of roles in HLPSL:
\begin{itemize}
    \item \textbf{Basic Role:} The execution behaviour of the principals is defined with the help of this role. In our protocol, we have two essential roles called as IoT user and application gateway. 
    \item \textbf{Composed Role:} Session modelling of fundamental role events are presented using this role. 
    \item \textbf{Environment Role:} Effective principles and sessions that must be considered during implementation are defined by this role. The starting point is considered as a highest role for simulation execution. 
\end{itemize}
The HLPSL protocol is translated into Intermediate Format (IF) using the HLPSL2IF translator. We are using Delev-Yao medium that works based on two fundamental operations called as \textit{SND} and \textit{RCV}. This is the only medium supported by AVISPA, for communication simulation.
We ran a simulation for the proposed protocol's three phases (registration phase, login phase, and key generation phase). 

The AVISPA tool comprises four sub-tools that function as the back-end to IF, as indicated in Figure \ref{avispa}. To obtain Output Format (OF), the Protocol transformed in IF is injected into these tools. the Tree Automata based on Automatic Approximation for the Analysis of Security Protocol tool, the SAT-Based Model Tracker tool, The On-the-Fly Model (OFM) Checker tool, and the Constraint Logic-based Model Checker tool are four major back end tools utilized for computing OF, which states whether or not the proposed protocol is secure against said attacks. 
\begin{figure}
    \centering
     \includegraphics[width=\linewidth]{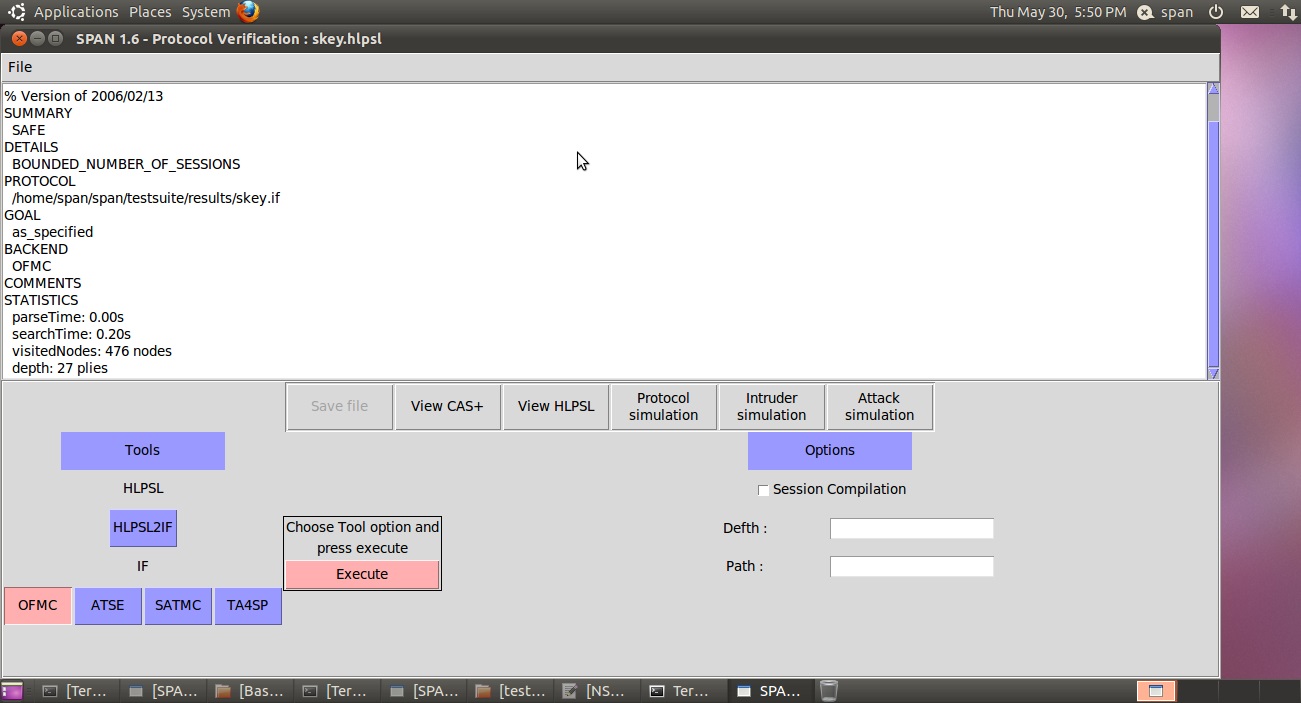}
    \caption{AVISPA OFMC tool Output}
    \label{avispa2}
\end{figure}

The simulation of the proposed protocol using the famed OFMC tool can be seen in the Figure \ref{avispa2}. We utilized two back-end tools (OFMC and CL-AtSe) to simulate the proposed protocol. Due to their inability to simulate bitwise XOR operations, SATMC and TA4SP cannot simulate the proposed protocol, hence whatever results these tools provide are inconclusive.   
\begin{figure}
    \centering
     \includegraphics[width=\linewidth]{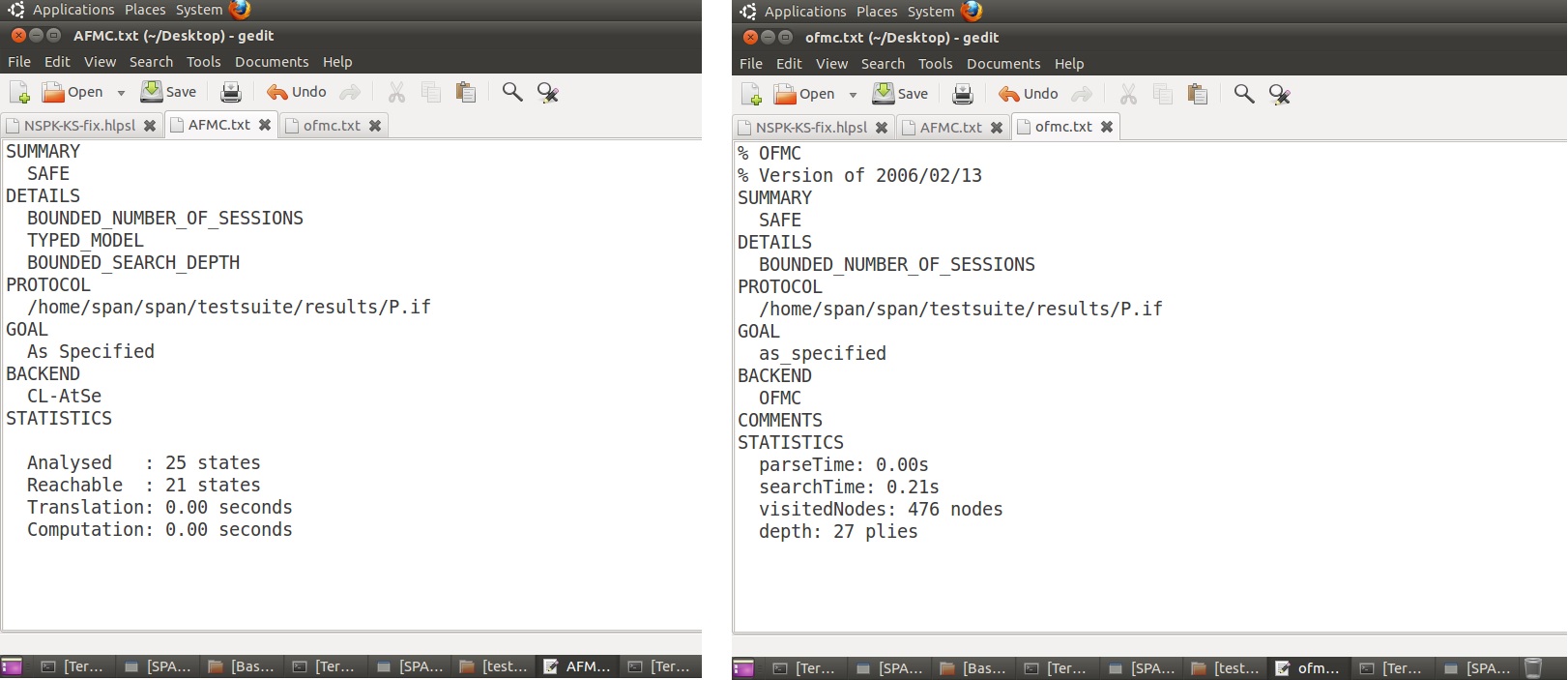}
    \caption{AVISPA Simulation Output}
    \label{avispa1}
\end{figure}

The protocol tested 25 states during simulation using the CL-AtSe tool, of which 21 were attainable, as shown in Figure \ref{avispa1}. While an OFMC simulation revealed that the simulation visited 476 nodes up to a depth of 27 heaps in 0.21 seconds. The protocol is secure against a replay attack and a Man-in-the-Middle attack, according to simulations employing both tools. 
\section{\textbf{Performance Analysis}}
\noindent In this section, we evaluate the developed protocol's performance in terms of computing and communication overhead.  The computation cost analyzes the proposed protocol in terms of the utilization of cryptographic operation. The communication cost analyzes protocol regarding the number of "bits" transmitted by each participating party. 
\label{sec:performancecomparison}
\subsection{\textbf{Computation cost analysis}}
\noindent The computation cost defines the total time consumed by the scheme for session key generation. In the proposed scheme, we uses different functions/methods for session key generation. Let us define this functions as, \textit{$T_{h}$}, \textit{$T_{Pa}$}, \textit{$T_{Pm}$}, and \textit{$T_{Syn}$}. That is time complexity for the hash computation, elliptic curve point addition operation, elliptic curve scalar multiplication operation and symmetric encryption and decryption operation respectively. We extracted the individual time required for each these functions. That was 0.0024 ms, 0.029 ms, 2.227 ms, 0.0047 ms for \textit{$T_{h}$}, \textit{$T_{Pa}$}, \textit{$T_{Pm}$}, \textit{$T_{Syn}$} respectively.  Table \ref{compcost} shows comparative analysis for the computation cost between the proposed scheme and other existing schemes and proves computational efficiency of the proposed scheme. 
\begin{table}
    \centering
    \caption{\textbf{Comparison of Computation Cost}}
    \begin{tabular}{p{1.5cm}p{5cm}} \hline 
    \textbf{Scheme} & \textbf{Computation Cost}  \\ \hline
    \small{\cite{14islam}} & 6\textit{$T_{Pm}$} + 1\textit{$T_{Pa}$} + 10\textit{$T_{h}$} $\approx$ 13.4078 \\
    \small{\cite{33Zhang}} & 6\textit{$T_{Pm}$} + 2\textit{$T_{Syn}$} + 11\textit{$T_{h}$} $\approx$ 13.3905 \\ 
    \small{\cite{23Odelu}} & 5\textit{$T_{Pm}$} + 3\textit{$T_{Syn}$} + 13\textit{$T_{h}$} $\approx$ 13.4767 \\    
   \small{\cite{21Liu}} & 6\textit{$T_{Pm}$} + 4\textit{$T_{Syn}$} + 11\textit{$T_{h}$} $\approx$ 13.3997 \\ 
    \small{\cite{26Qiu}} & 13$T_h$ + 4$T_{Pm}$ $\approx$ 8.9339 ms \\
    \small{Proposed} & 2\textit{$T_{Pm}$} + 5\textit{$T_{Syn}$} + 11\textit{$T_{h}$} $\approx$ 4.5003 \\  \hline
    \end{tabular}
    \label{compcost}
\end{table}

\subsection{\textbf{Communication cost analysis}}
\noindent The communication cost shows the total number of bits transmitted before establishing the key over a channel for authentication. We computed the size of individual parameters (in bits) for communication cost computations. The size of the user identity is 160 bits; the output of the hash operation is 160 bits, the size of the randomly generated nonce is 128 bits, and the size of the time-stamp is 32 bits. We use ECC in the proposed protocol. Each point in the elliptic curve is built up using two coordinates, each with 160 bits; hence, the total size of the point  (Xp, Yp) is 320  bits. The public key size ($UPW*P$)  is 320, and the private key ($UPW$) size is 160 bits as per  ECC computations. Table \ref{commcost} presents a analogous analysis of the communication costs for our protocol with other available protocols for the similar environment.
\begin{table}
   \centering
   \caption{\textbf{Comparison of Communication Cost}}
    \begin{tabular}{p{1.5cm}p{3cm}p{2.5cm}} \hline
    \textbf{Scheme} & \textbf{No of Message} & \textbf{Comm. cost}  \\ \hline
    \small{\cite{14islam}} & 2 & 1184 bits \\ 
    \small{\cite{33Zhang}} & 2 & 1184 bits \\ 
    \small{\cite{23Odelu}} & 2 & 1600 bits \\ 
    \small{\cite{21Liu}} & 2 & 1952 bits \\
    \small{\cite{26Qiu}} & 3 & 1280 bits \\ 
    \small{Proposed} & 2 & 992 bits \\ \hline 
    \end{tabular}
    \label{commcost}
\end{table}
\section{Implementation}
\label{sec:implement}
\noindent For implementation, we used two types of user devices and two types of gateway devices. We used five nodeMCU and five Raspberry-Pis as light-weight user devices and two laptops as resource-capable user devices. As a resource-capable gateway device, we used a laptop with a configuration of 8 GB RAM, Intel (R) Core (TM)i7-5500U CPU, 2.40 GHz, 64 bit, Ubuntu 16.04 operating system. As a lightweight gateway device, we used Raspberry Pi 3 Model B with Quad Core 1.2GHz Broadcom BCM2837 64bit CPU, 1GB RAM with BCM43438 wireless LAN, and Bluetooth Low Energy  (BLE) on board. As a programming language, we used python 3. We used publish-subscribe-based MQTT Protocol as a communication protocol at the application layer, which uses TCP at the transport layer protocol and 6LoWPaN at the network layer.

We implemented all basic elliptic curve operations in python language by connecting Raspberry Pi as a gateway device and the desktop system as a user device. 
\subsection{\textbf{Computed session key}}
\noindent Following Figure \ref{fig : 3}. shows computed session key using proposed work:
\begin{figure}
    \centering
     \includegraphics[width=\linewidth]{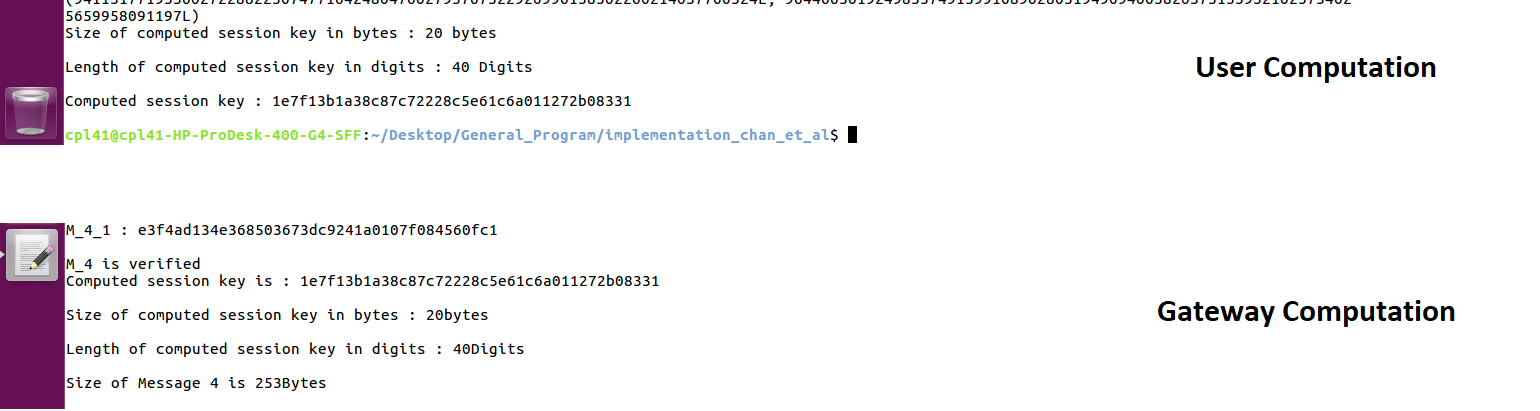}
    \caption{Computed Session key}
    \label{fig : 3}
\end{figure}
\subsection{\textbf{Networking parameters}}
\noindent We computed individual networking parameters, such as round trip delay, packet loss, and throughput. Implementing the proposed scheme using ECC operations and the MQTT protocol reduces the round trip delay, increases throughput, and reduces packet loss for the proposed scheme. We used the widely adopted packet sniffing and packet analyzing tool "Wireshark" to compute the aforementioned parameters. The Wireshark captures MQTT packets and provides all essential networking parameters in an unstructured format; thus, we used python programming to retrieve data more structured way. We interpret these parameters using an automated python script that reads Wireshark data and evokes the required information. 
\section{Conclusion and Future Work}
\label{sec:conclusion}
\noindent For the U-GW-based generic IoT paradigm, this article developed a unique, reliable, and lightweight SC-based remote user authentication technique.
For the discrete logarithm operations, we employed ECC. We performed an informal security study utilising the Dolev-Yao channel and a formal security analysis using the widely used BAN Logic and AVISPA tools for the proposed method. We show that the suggested scheme is very efficient and reliable in terms of communication cost, computation cost, end-to-end delay, packet loss, and throughput for the U-GW system model by comparing it to other current methods. We utilised the Raspberry Pi and the MQTT protocol for implementation. We're also working on a multi-factor model to overcome the restrictions of two-factor authentications and also on U-GW-sensor model with lighter implementation using protocols like light-MQTT and Bluetooth low energy as a future work for the proposed scheme. 

\bibliographystyle{IEEEtran}
\bibliography{mybibfile}
\end{document}